\documentclass[11pt]{article}

\usepackage{amsfonts,amsmath,latexsym,color,epsfig,hyperref,enumitem, amssymb}
\newtheorem{theorem}{Theorem}

\newtheorem{lemma}{Lemma}
\newtheorem{proposition}{Proposition}

\newenvironment{proof}
      {\medskip\noindent{\bf Proof:}\hspace{1mm}}
      {\hfill$\Box$\medskip}
\def\qed{\ifvmode\mbox{ }\else\unskip\fi\hskip 1em plus 10fill$\Box$}
\newenvironment{proofof}[1]
      {\medskip\noindent{\bf Proof of #1:}\hspace{1mm}}
      {\hfill$\Box$\medskip}
\input{epsf}
\makeatletter
\def\Ddots{\mathinner{\mkern1mu\raise\p@
\vbox{\kern7\p@\hbox{.}}\mkern2mu
\raise4\p@\hbox{.}\mkern2mu\raise7\p@\hbox{.}\mkern1mu}}
\makeatother

\title{\vspace{-0.7cm} On the trifference problem for linear codes}
\author{Cosmin Pohoata\thanks{Department of Mathematics, Yale University, USA. Email: {\tt andrei.pohoata@yale.edu}.} \and Dmitriy Zakharov\thanks{Laboratory of Combinatorial and Geometric Structures, MIPT, Russia. Email: {\tt  zakharov2k@gmail.com}.}}
\date{}

\begin{document}
\maketitle

\begin{abstract}
We prove that perfect $3$-hash linear codes in $\mathbb{F}_{3}^{n}$ must have dimension at most $ \left(\frac{1}{4}-\epsilon\right)n$ for some absolute constant $\epsilon > 0$.
\end{abstract}

\section{Introduction}

A code of length $n$ over an alphabet of size $k \geq 3$ is a subset $\mathcal{F}$ of $\left\{0,1,\ldots,k-1\right\}^n$. Such a code is called a {\textit{perfect $k$-hash code}} if for every subset of $k$ distinct elements (or {\textit{codewords}}) of $\mathcal{F}$, say $\left\{c^{(1)},\ldots,c^{(k)}\right\}$, there exists a coordinate $i$ such that all these elements differ in this coordonate, namely $\left\{c^{(1)}_{i},\ldots,c^{(k)}_{i}\right\}=\left\{0,\ldots,k-1\right\}$. The problem of finding upper bounds for the maximum size of perfect $k$-hash codes is a fundamental problem in theoretical computer science, with a long history. The original benchmark result comes from an elementary double counting argument, which gives the bound
\begin{equation} \label{korner}
|\mathcal{F}| \leq (k-1) \cdot \left(\frac{k}{k-1}\right)^{n},
\end{equation}
for all $k \geq 3$. In 1973, K\H{o}rner \cite{Ko73} raised the beautiful problem of possibly improving upon this simple and yet surprisingly strong inequality (which for example defies all known algebraic techniques, including the recent slice-rank method of Tao \cite{Tao}; see for example \cite{CD19} for a discussion). In 1984, this was famously (partially) achieved by Fredman and Koml\'os, who managed to improve on this bound for all $k \geq 4$ (and for sufficiently large $n$), by showing that 
\begin{equation} \label{FK}
|\mathcal{F}| \leq \left(2^{k!/k^{k-1}}\right)^{n}.
\end{equation}
Their argument also carried through without any issues for $k=3$, however, in this case note that $2^{2/3} \approx 1.5874$, so the upper bound from \eqref{FK} provided a much weaker inequality than $|\mathcal{F}| \leq 2 \cdot (3/2)^n$. Consequently, short after, the problem of breaking the $3/2$ barrier for perfect $3$-hash codes in $\left\{0,1,2\right\}^n$ quickly established itself as one of the most tantalizing problems in information theory and became known as the so-called Trifference Problem of K\H{o}rner (perfect $3$-hash codes are also called {\textit{trifferent sets}}). 

Over the years, there have been multiple interesting refinements of the Fredman-Koml\'os method which led to various slight improvements of the inequality from \eqref{FK} for $k \geq 4$. See for example \cite{DGR17} and \cite{CD20} for some recent progress on the cases $k=4$ and $k\geq 5$, respectively. We also encourage the interested reader to consult \cite{GR19}, which is the paper that \cite{CD20} builds upon and which also contains a wonderful comprehensive discussion of the history of the problem and of the Fredman-Koml\'os approach. That being said, despite all this activity, no progress of any kind has ever been made for the case $k=3$, where the best known upper bound remained the stubborn inequality from \eqref{korner}. 

%The {\textit{rate of the code}} $\mathcal{F}$ (in bits) is typically defined as $R = \log_{2}{|\mathcal{F}|}/n$. For fixed positive integer $k$, 

%Such sets are called {\textit{perfect $3$-hash codes}} or {\textit{trifferent sets}}. 

In this paper, we prove threshold-breaking bounds for perfect $3$-hash {\textit{linear}} codes. We say that $V \subset \mathbb{F}_{3}^{n}$ is a perfect $3$-hash linear code if $V$ is a subspace of $\mathbb{F}_{3}^{n}$ which has the trifference property that for each triplet of distinct vectors $x,y,z$ in $V$ there exists some coordinate $i$ where $\left\{x_i,y_i,z_i\right\}=\left\{0,1,2\right\}$. 

\begin{theorem} \label{main2}
There exists a positive integer $n_{0}$ such that for every $n \geq n_{0}$, every perfect $3$-hash linear code in $\mathbb{F}_{3}^{n}$ has dimension at most $\left(\frac{1}{4}-\epsilon\right)n$, for some absolute constant $\epsilon > 0$.
\end{theorem}

Linear perfect $3$-hash codes are not only of independent interest but K\H{o}rner and Marton also use them to construct the largest known examples of trifferent sets. These are sets of size $(9/5)^{n/4} \approx 1.158^n$ and rely on the existence of $2$-dimensional linear perfect $3$-hash codes in $\mathbb{F}_{3}^{4}$. See \cite{KM88} for more details.

To prove Theorem \ref{main2}, we use a mix of probabilistic and algebraic arguments, together with a series of combinatorial reductions. Some of these reductions have led us to some interesting separate problems that we decided to single out, as they may be interesting for independent reasons. 

\section{Proof of Theorem \ref{main2}}

To set things up, we shall first prove the following preliminary bound, which already goes below the $3/2$ barrier from \eqref{korner} for linear codes.

\begin{proposition}\label{lin}
Let $V \subset \mathbb F_3^n$ be a $d$-dimensional $3$-hash code. Then, 
$$d \le \frac{n+11}{4}.$$
\end{proposition}

Before we prove Proposition \ref{lin}, let us give an equivalent formulation. For a vector subspace $V \subset \mathbb F_3^n$ we denote by $V^*$ the space of linear functions on $V$ and let $x_1, \ldots, x_n \in V^*$ be the coordinate functions on $\mathbb F_3^n$ restricted to $V$. That is, a point $(a_1, \ldots, a_n) \in V$ is mapped by $x_i$ to $a_i$. Then the ``trifference hypothesis" applied to vectors $0, v, -u \in V$ implies that there is $i$ such that $x_i(v) = 1$ and $x_i(-u) = -1$ or $x_i(v) = -1$ and $x_i(-u) = 1$. In other words, for any two affine hyperplanes $H_1, H_2 \subset V^*$ given by equations $H_1 = \{x ~|~ x(v) = 1 \}$ and $H_2 = \{x ~|~ x(u) = 1\}$ there is $i$ such that either $x_i \in H_1 \cap H_2$ or $-x_i \in H_1 \cap H_2$. So we arrive at the following problem:

\bigskip

{\bf{Alternate Problem 1}}. For $d\ge 1$, what is the minimum size $f(d)$ of a centrally symmetric set $X \subset \mathbb F_3^d \setminus \{0\}$ such that for any two non-parallel hyperplanes $H_1, H_2 \subset \mathbb F_3^d$ which do not contain 0 there is $x \in X$ which belongs to $H_1 \cap H_2$?

\bigskip

Let $Y$ denote the set of linear functions $\{x_1, \ldots, x_n\}$ on the space $V$ and let $X = Y \cup (-Y) \setminus \{0\}$. Then clearly $X$ is a centrally symmetric set of size at most $2n$. Moreover, as we showed above, $X$ satisfies the condition of Alternate Problem 1. Thus, we have the lower bound $2n \ge f(d)$.
So Proposition \ref{lin} will follow from the following inequality.

\begin{proposition}\label{lin2}
Under the notations from Alternative Problem 1, we have 
$$f(d) \ge 8d - 22.$$
\end{proposition}

For the proof of Proposition \ref{lin2}, we will need the following preliminary result due to Brouwer and Schrijver \cite{BS78}.

\begin{lemma}\label{AF}
For any set $X \subset \mathbb F_3^d$ of size at most $2d$ there is an affine hyperplane $H \subset \mathbb F_3^d$ which does not intersect $X$.
\end{lemma}

For the reader's convenience, we include a short proof using the (quantitative) version of Alon's Combinatorial Nullstellensatz \cite{NA99} due to Karasev and Petrov \cite{KP12}.

\begin{proofof}{Lemma \ref{AF}}
We denote by $\langle \cdot, \cdot \rangle$ the standard scalar product on $\mathbb F_3^d$. Let $X \subset \mathbb F_3^d$ be a set of size at most $2d$ and consider the following polynomial in $d+1$ variables $t_1, \ldots, t_{d+1}$:
$$
P(t_1, \ldots, t_{d+1}) = \prod_{x \in X} (\langle x, t\rangle - t_{d+1}),
$$
where $t:=(t_1, \ldots, t_d)$. If $t \neq 0$ then it defines a hyperplane
$$
H_{t, t_{d+1}} = \{x \in \mathbb F_3^d ~|~ \langle x, t\rangle = t_{d+1}\}.
$$
Therefore, 
$$P(t_1,\ldots,t_{d+1}) = 0\ \ \text{if and only if}\ \ H_{t, t_{d+1}}\ \text{intersects}\ X.$$
So we may assume that $P(t_1,\ldots,t_{d+1}) = 0$ holds
for all vectors $(t_1,\ldots,t_{d+1})$ whose first $d$ coordinates do not all vanish. However, note that $P(0, \ldots, 0) = 0$ and $P(0, \ldots, 0, 1) = (-1)^{d+1}$, so at this point we can invoke the quantitative version of the Combinatorial Nullstellensatz \cite[Theorem 4]{KP12}. 

\begin{lemma} \label{CN}
Let $d_{1},\ldots,d_{n}$ be non-negative integers and let $P(x_1,x_2,\dots,x_n)$ be a $n$-variate polynomial of degree at most $d_1+\ldots+d_n$ with coefficients in some field $\mathbb{F}$. Denote by $C$ the
coefficient at $x_1^{d_1}\dots x_n^{d_n}$ in $P$ (which may possibly be $0$) and let $U_1$, $U_2$, \ldots, $U_n$ be arbitrary subsets
of $\mathbb{F}$ such that $|U_i|=d_i+1$ for every $i=1,\ldots,n$.
Furthermore, let $D(U,\alpha)=\prod_{c\in U\setminus \{\alpha\}}
(\alpha-c)$ for $\alpha\in U$. 
Then,
$$
C=\sum_{\alpha_1\in U_1} \ldots \sum_{\alpha_{n} \in U_{n}} \frac{P(\alpha_1,\dots,\alpha_n)}
{D(U_1,\alpha_1)\dots D(U_n,\alpha_n)}.$$
\end{lemma}

Applied when $n=d+1$, $\mathbb{F} = \mathbb{F}_{3}$ and with $U_{1}=\ldots=U_{d}=\mathbb{F}_{3}$, $U_{d+1}=\left\{0,1\right\}$, Lemma \ref{CN} ensures that the coefficient of $P$ in front of the monomial $t_1^2 \ldots t_d^2 t_{d+1}$ is not zero which means that the degree of $P$ must be at least $2d+1$. This is clearly a contradiction since, by definition, $P$ was a product of $|X| \leq 2d$ linear terms. It thus follows that there must be some $t$ such that $H_t \cap X = \emptyset$. 
\end{proofof}

In order to obtain a lower bound on $f(d)$ we will consider yet another auxiliary problem.

\bigskip

{\textbf{Alternate Problem 2}}. For given integers $d, n \ge 1$, what is the maximum number $m = m(n, d)$ such that for every centrally symmetric set $X \subset \mathbb F_3^d \setminus \{0\}$ of size $n$ there is a hyperplane passing through zero and such that $|X \cap H| \ge m$?

\bigskip

Note that $m(n, d)$ is monotonically increasing in $n$.
We can bound $f(d)$ in terms of the function $m(n, d)$ as follows.

\begin{lemma}\label{aux1}
If for given $n, d \ge 1$ we have $m(n, d) \ge n - 4d+4$, then 
$$f(d) \ge n+2.$$
\end{lemma}

\begin{proofof}{Lemma \ref{aux1}}
Let $X \subset \mathbb F_3^d$ be a centrally symmetric set of size $n$ and $m(n, d) \ge n-4d+4$. Then by the definition of $m$ applied to $X$ there is a hyperplane $H_1 \subset \mathbb F_3^d$ such that $|X \cap H_1| \ge n - 4d+4$ and so the size of the set $X \setminus H_1$ is at most $4d-4$. By the pigeonhole principle there is an affine hyperplane $H'_1$ parallel to $H_1$ such that the set $Y = X \cap H'_1$ has size at most $2d-2$. Pick any point $x \in -H_1'$ and apply Lemma \ref{AF} to the set $Y \cup \{0\} \cup \{x\}$.  
So there is a hyperplane $H_2$, not passing through zero and not parallel to $H'_1$ such that $Y \cap H_2 = \emptyset$. 

So we constructed a pair of non-parallel hyperplanes $H'_1$ and $H_2$ such that $X \cap H'_1 \cap H_2 = \emptyset$. This proves that $f(d) > n$. But $f(d)$ and $n$ are even numbers and so we in fact have $f(d) \ge n+2$.

\end{proofof}

In the context of Alternate Problem 2, if we take $H$ at random then from the linearity of expectation we immediately obtain 
$$
m(n, d) \ge \frac{3^{d-1}-1}{3^d-1}n \ge \frac{n}{3} - \frac{n}{3^{d-1}}.
$$
So by Lemma \ref{aux1}, if $n/3 \ge n-4d+4 + \frac{n}{3^{d-1}}$, then $f(d) \ge n+2$. In other words, we already proved that $f(d) \ge 6d-10$.
Now we improve this trivial lower bound on $m$ to wrap up the proof of Proposition \ref{lin2}.

\begin{lemma}\label{lm}
For any integers $d$ and $n$ such that $d \ge 3$ and $n \ge 2d$, 
$$m(n, d) \ge \frac{n + 4d}{3} - 3 - \frac{n}{d}.$$
\end{lemma}

\begin{proofof}{Lemma \ref{lm}}
Let $X \subset \mathbb F_3^d \setminus \{0\}$ be a centrally symmetric set of size $n$. We may clearly assume that $X$ is not contained in any hyperplane which passes through zero. So we can find vectors $e_1, \ldots, e_d \in X$ which form a basis of $\mathbb F_3^d$. Moreover, we may assume that $(e_1, \ldots, e_d)$ is the standard basis of $\mathbb F_3^d$. Note that vectors $-e_1, \ldots, -e_d$ also belong to $X$ since $X$ is centrally symmetric. Pick uniformly at random a pair $\{i, j\} \subset {[d] \choose 2}$ and consider the hyperplane 
$$
H = H_{i, j} = \{ (v_1, \ldots, v_d) \in \mathbb F_3^d~|~ v_i = v_j\}.
$$
Let us estimate the expected size of the intersection $|H \cap X|$. Put $E = \{\pm e_1, \ldots, \pm e_d\} \subset X$ and $Y = X \setminus E$. Note that $|H \cap E| = 2d-4$ for any $i, j$. Now consider a fixed vector $w \in \mathbb F_3^d$ and let us find the probability $p_w$ that $w \in H$. Suppose that $w$ has precisely $a$, $b$ and $c$ coordinates equal to $-1,0, 1 \in \mathbb{F}_{3}$, respectively. First, note that since $w$ is $d$-dimensional, we must have that $a+b+c=d$. Second, note that 
$$
p_w = \frac{a(a-1)+b(b-1)+c(c-1)}{d(d-1)} \ge \frac{d^2/3 - d}{d(d-1)} = \frac{1}{3} - \frac{2}{3d-3},
$$
where in the first inequality we used the simple fact that $3(a^2+b^2+c^2) \geq (a+b+c)^2 = d^2$. This allows us to conclude that
\begin{eqnarray*}
\mathbb E \left[|H \cap X|\right] &=& \mathbb E \left[|H \cap E|\right] +\mathbb E \left[|H \cap Y|\right] \\
&\ge& 2d-4 + \left(\frac{1}{3} - \frac{2}{3d-3}\right)|Y| \\
&\ge& 2d-4 + \frac{n-2d}{3} - \frac{2n - 4d}{3d-3}.
\end{eqnarray*}
After some further simple manipulations, this yields
\begin{equation*}
    \mathbb E\left[|H \cap X|\right] \ge \frac{n + 4d}{3} - 3 - \frac{n}{d},
\end{equation*}
which means that there exists a choice of $\{i, j\}$ such that $H_{i, j}$ satisfies $|H_{i, j} \cap X| \ge \frac{n + 4d}{3} - 3 - \frac{n}{d}$, which settles the proof of Lemma \ref
{lm}.
\end{proofof}

\bigskip

Returning to the proof of Proposition \ref{lin2}, we now let $n = 8d - a$; then, by Lemma \ref{lm} we have 
$$
m(n, d) \ge \frac{12d - a}{3} - 3 - 8 = 4d - 11 - \frac{a}{3}
$$
and so 
$$
m(n, d) - (n-4d+4) \ge (4d-11 - \frac{a}{3}) - 4d - 4 + a= \frac{2a}{3} - 15.
$$
Therefore, if we take $a = 24$ then $\frac{2a}{3} - 15 > 0$. By Lemma \ref{aux1}, we therefore get
$$
f(d) \ge n+2 = 8d - 22,
$$
which completes the proof of Proposition \ref{lin2}, and thus that of Proposition \ref{lin}. 

\subsection{An improved lower bound on $m(n, d)$}

The stronger bound in Theorem \ref{main2} follows essentially from a better estimate on $m(n, d)$, which we will present below. We start with a simple packing bound.

\begin{proposition}\label{pc}
Let $V \subset \mathbb F_3^n$ be a $d$ dimensional subspace such that any non-zero vector $v \in V$ has at least $2k+1$ non-zero coordinates. Then we have
\begin{equation}\label{pack}
    {n \choose k} 2^k \le 3^{n-d}.
\end{equation}
\end{proposition}

\begin{proof}
Denote by $B \subset \mathbb F_3^n$ the set of vectors $v$ such that $v$ has at most $k$ non-zero coordinates. Then we have $|B| \ge {n \choose k} 2^k$ and it is easy to see that $b_1 + v_1 \neq b_2 + v_2$ for any distinct pairs $(b_1, v_1), (b_2, v_2) \in B \times V$. We conclude that 
$$
3^n \ge |B| |V| \ge {n \choose k} 2^k 3^d.
$$
\end{proof}

We use this to deduce the following preliminary result.

\begin{lemma}\label{mb}
Let $m = m(2n, d)$ for some $n \ge d \ge 1$ and let $k = \lfloor\frac{2n-m-1}{4}\rfloor$. Then, 
$${n \choose k} 2^k \le 3^{n-d}.$$
\end{lemma}

\begin{remark}
One can check that the bound on $m(n, d)$ which follows from Lemma \ref{mb} beats the bound from Lemma \ref{lm} for $n \lesssim 5.5d$.
\end{remark}

\smallskip

\begin{proofof}{Lemma \ref{mb}}
Let $X = \{x_1, -x_1, x_2, -x_2, \ldots, x_n, -x_n\}$ be a subset of $V = \mathbb F_3^d$ such that any hyperplane contains at most $m-1$ elements of $X$. Note that, in particular, the set $X$ is not entirely contained in any hyperplane.

Define a linear map $\varphi: V^* \rightarrow \mathbb F_3^n$ which sends a linear function $\xi$ on $V$ to the vector $(\xi(x_1), \ldots, \xi(x_n))$. Since $X$ is not contained in any hyperplane, the image of $\varphi$ in $\mathbb F_3^n$ is a $d$-dimensional subspace $U \subset \mathbb F_3^n$. We claim any vector $\varphi(\xi) \in U$ has at least $n-\frac{m-1}{2}$ non-zero coordinates. Indeed, if the number of zero coordinates of the vector $\varphi(\xi) \in U$ is equal to $t$ then the hyperplane $\{\xi = 0\}$ contains exactly $2t$ elements of $X$. So $2t \le m-1$ by the assumption on $X$. By applying Proposition \ref{pc} with $k = \lfloor \frac{2n-m-1}{4}\rfloor$, we obtain the desired inequality.
\end{proofof}

\begin{lemma}\label{tech}
There is an absolute constant $\varepsilon > 0$ such that for every sufficiently large $d$ and positive integer $n$ satisfying $8d \le n \le (8+\varepsilon)d$, we have 
$$m(n, d) \ge n-4d+4.$$
\end{lemma}

Note that Lemma \ref{tech} implies Theorem \ref{main2}. Indeed, by Lemma \ref{aux1} and Lemma \ref{tech} we have $f(d) \ge (8+\varepsilon)d$. But as we observed during the proof of Proposition \ref{lin}, if there is a $d$-dimensional linear perfect 3-hash code in $\mathbb F_3^n$ then $2n \ge f(d) \ge (8+\varepsilon)d$. So in order to prove Theorem \ref{main2} it is enough to prove Lemma \ref{tech}. 

\smallskip

\begin{proofof}{Lemma \ref{tech}}
Let $\alpha \in [0, 1]$ and denote $n = (4+\alpha)d$.
Let $X \subset \mathbb F_3^d \setminus \{0\}$ be a centrally symmetric set of size $2n$ such that $|H \cap X| < m = m(2n, d)$ for any hyperplane $H$. Arguing indirectly, we assume that $m(2n, d) \le 2n-4d+3$. We are going to derive a contradiction provided that $\alpha$ is small enough.  

Note that by Lemmas \ref{aux1} and \ref{lm} we have
$$
\frac{2n+4d}{3}-13 \le m \le 2n-4d+3,
$$
since $n = (4+\alpha)d$ this simplifies to
\begin{equation}
    m-4d \in \left[\frac{2\alpha}{3}d-5, 8\alpha d+3\right].
\end{equation}

Let $H_0 \subset \mathbb F_3^d$ be a hyperplane such that $|X \cap H_0| = m-1$. The existence of such a hyperplane follows from the minimality of $m$ in the definition of $m(n, d)$.

Let $Y = X \cap H_0$. This is a centrally symmetric subset of a $(d-1)$-dimensional space of size $2n' = m-1$. So if we denote $m' = m(2n', d-1)$ then by Lemma \ref{mb} we have
\begin{equation}\label{mess}
    {n' \choose \lfloor\frac{2n'-m'-1}{4}\rfloor} 2^{\lfloor\frac{2n'-m'-1}{4}\rfloor} \le 3^{n' - d+1}.
\end{equation}
Let $H_1 \subset H_0$ be a hyperplane in $H_0$ such that $|Y \cap H_0|\ge m'$. By the pigeonhole principle there is a hyperplane $H_2 \subset \mathbb F_3^d$ such that $H_2 \cap H_0 = H_1$ and $|H_2 \cap (X \setminus Y)| \ge |X\setminus Y|/3$. We conclude that
\begin{equation}\label{mm}
    m \ge |H_2 \cap X| = |H_2 \cap Y| + |H_2 \cap (X \setminus Y)| \ge |Y| + \frac{2n-|Y|}{3} \ge m' + \frac{2n-m}{3} = \frac{2n + 3m'-m}{3}.
\end{equation}
Now we study (\ref{mess}). After plugging in $n' = (m-1)/2$ and removing floors we get
\begin{equation}\label{ugly}
    {m/2 \choose \frac{m-m'}{4}} 2^{\frac{m-m'}{4}} \le c 3^{m/2 - d},
\end{equation}
for some reasonable constant $c \ge 1$. 

Note that $m' = m(m-1, d-1)$ and so $m' \ge m/3$. Therefore, $x = \frac{m-m'}{4} \le m/6$. The function ${m/2 \choose x}2^x$ is increasing in $x$ on the interval $[0, m/6]$. So if we increase $m'$ then the left hand side in (\ref{ugly}) decreases. Thus, by (\ref{mm}), the bound (\ref{ugly}) holds for $m' = \frac{4m - 2n}{3}$:
\begin{equation}\label{mnm}
    {m/2 \choose \frac{2n - m}{12}} 2^{\frac{2n-m}{12}} \le c 3^{m/2 - d}.
\end{equation}
Denote $m = (4+\gamma)d$ and recall that $\gamma \in [\frac{2\alpha}{3}, 8\alpha]$ (up to $O(\frac{1}{d})$). After some simple manipulations we get
\begin{equation}\label{Oa}
    {(2+O(\alpha))d \choose (\frac{1}{3}+O(\alpha))d} 2^{(\frac{1}{3}+O(\alpha))d} \le 3^{(1 + O(\alpha))d},
\end{equation}
but ${2d \choose d/3} 2^{d/3} \approx 3.10^d$. This implies that there is some $\varepsilon > 0$ such that if $\alpha \le \varepsilon$ then the left hand side of (\ref{Oa}) is, say, at least $3.05^d$ and the right hand side is at most $3.04^d$. This is a contradiction, and so we must have that $n \ge (4+\varepsilon)d$ for some absolute constant $\varepsilon>0$.

\end{proofof}

{\bf{Acknowledgements}}. The first author would like to thank Michail Sarantis and Prasad Tetali for several helpful discussions. We would also like to thank Marco Dalai for bringing reference \cite{BS78} to our attention. Finally, the second author would also like to acknowledge the support of the grant of the Russian Government N 075-15-2019-1926.

\end{document}